\documentclass[11pt,a4paper]{article}

\usepackage[vmargin=25mm,hmargin=33mm]{geometry}
\usepackage{amsmath,amsthm,amssymb,amsfonts,complexity,enumitem,graphicx,url,doi,bm}
\usepackage[numbers,sort&compress]{natbib}
\usepackage[table]{xcolor}
\usepackage{tabularx,hhline}
\usepackage[small]{caption}
\usepackage{microtype}
\usepackage{hyperref}

\definecolor{citecolor}{HTML}{0000C0}
\definecolor{urlcolor}{HTML}{000080}

\hypersetup{
    colorlinks=true,
    linkcolor=black,
    citecolor=citecolor,
    filecolor=black,
    urlcolor=urlcolor,
    pdfauthor={Pierre Fraigniaud, Mika G{\"o}{\"o}s, Amos Korman, Jukka Suomela},
    pdftitle={What can be decided locally without identifiers?},
}

\setitemize{label=$-$}

\newtheorem{theorem}{Theorem}
\newtheorem{corollary}{Corollary}
\newtheorem{lemma}{Lemma}

\newclass{\LCL}{LCL}
\newclass{\LD}{LD}
\newclass{\NLD}{NLD}
\newclass{\BPLD}{BPLD}
\newclass{\OI}{OI}
\newclass{\PO}{PO}

\newcommand{\local}{\mathcal{LOC\mspace{-3mu}AL}}
\newcommand{\inp}{\ensuremath{\boldsymbol x}}
\newcommand{\Id}{\mbox{\rm Id}}
\newcommand{\Pp}{\mathcal{P}}
\newcommand{\Rr}{\mathcal{R}}
\newcommand{\cC}{\mathcal{C}}
\newcommand{\cH}{\mathcal{H}}
\renewcommand{\C}{\mathbf{C}}
\newcommand{\B}{\mathbf{B}}

\newcommand{\yes}{\textit{yes}}
\newcommand{\no}{\textit{no}}

\newenvironment{myabstract}
               {\list{}{\listparindent 1.5em%
                        \itemindent    \listparindent
                        \leftmargin    0pt
                        \rightmargin   0pt
                        \parsep        0pt}%
                \item\relax}
               {\endlist}

\newenvironment{mycover}
               {\list{}{\listparindent 0pt
                        \itemindent    \listparindent
                        \leftmargin    0pt
                        \rightmargin   0pt
                        \parsep        0pt}%
                \raggedright
                \item\relax}
               {\endlist}

\begin{document}

\vspace*{2ex}
\begin{mycover}
{\LARGE \textbf{What can be decided locally without identifiers?}}

\bigskip
\bigskip
\textbf{Pierre Fraigniaud}%
\footnote{Additional support from the ANR projects DISPLEXITY, and from the INRIA project GANG.}\\
{\small CNRS and University Paris Diderot, France\\
\nolinkurl{pierre.fraigniaud@liafa.univ-paris-diderot.fr}\par}
\medskip

\textbf{Mika G\"o\"os}%
\\
{\small Department of Computer Science, University of Toronto, Canada\\
\nolinkurl{mika.goos@mail.utoronto.ca}\par}
\medskip

\textbf{Amos Korman}%
\footnotemark[1]\\
{\small CNRS and University Paris Diderot, France\\
\nolinkurl{amos.korman@liafa.univ-paris-diderot.fr}\par}
\medskip

\textbf{Jukka Suomela}%
\footnote{This work was supported in part by the Academy of Finland, Grants 132380 and 252018, and by the Research Funds of the University of Helsinki.}\\
{\small Helsinki Institute for Information Technology HIIT,\\Department of Computer Science, University of Helsinki, Finland\\
\nolinkurl{jukka.suomela@cs.helsinki.fi}\par}
\end{mycover}
\bigskip
\begin{myabstract}
\noindent\textbf{Abstract.}
Do unique node identifiers help in deciding whether a network $G$ has a prescribed property $\Pp$? We study this question in the context of \emph{distributed local decision}, where the objective is to decide whether $G\in\Pp$ by having each node run a constant-time distributed decision algorithm. If $G\in\Pp$, all the nodes should output \emph{yes}; if $G\notin\Pp$, at least one node should output \emph{no}.

A recent work (Fraigniaud et al., OPODIS 2012) studied the role of identifiers in local decision and gave several conditions under which identifiers are not needed. In this article, we answer their original question. More than that, we do so under all combinations of the following two critical variations on the underlying model of distributed computing:
\begin{itemize}[label=$-$,noitemsep]
\item ($\B$): the size of the identifiers is \emph{bounded} by a function of the size of the input network; as opposed to ($\neg \B$): the identifiers are \emph{unbounded}.
\item ($\C$): the nodes run a \emph{computable} algorithm; as opposed to ($\neg \C$): the nodes can compute any, possibly \emph{uncomputable} function.
\end{itemize}
While it is easy to see that under ($\neg \B,\neg \C$) identifiers are not needed, we show that under all other combinations there are properties that can be decided locally if and only if identifiers are present. Our constructions use ideas from classical computability theory.
\end{myabstract}
\bigskip
\begin{mycover}
\textbf{Keywords:} Distributed complexity; local decision; identifiers; computability theory.
\end{mycover}

\thispagestyle{empty}
\setcounter{page}{0}
\newpage

\section{Introduction}

\enlargethispage{3pt} 
In this work we ask and answer a simple question: \emph{Do we need unique node identifiers when locally deciding a graph property?} While this question is a natural one, our answers are somewhat artificial---but only necessarily so.

\paragraph{Local decision.}
A property of graphs $\Pp$ is \emph{locally decidable} if there is a distributed algorithm~$A$ (in the usual $\local$ model; see Section~\ref{ssec:local-model}) with a constant running time $t=O(1)$ that when run on a graph $G$ can decide whether $G\in\Pp$ in the following sense:
\begin{itemize}[noitemsep]
\item if $G\in\Pp$, then $A$ outputs \emph{yes} on every node of $G$, and
\item if $G\notin\Pp$, then $A$ outputs \emph{no} on at least one node of $G$.
\end{itemize}

Here, the output of $A$ on a node $v\in V(G)$ can only depend on the information that is available to within $t$ steps of $v$ in $G$. This includes not only the radius-$t$ neighbourhood topology around $v$, but also---as is often assumed---numerical identifiers $\Id(u)$ for each node $u$ in the neighbourhood. The assignment $\Id\colon V(G)\to\mathbb{N}$ is one-to-one.

\paragraph{Do we need identifiers?}
Recently, \citet{FHK12} asked whether it makes any difference in this context to have $A$'s output depend on the identifiers $\Id(v)$. After all, whether $G$ has the property $\Pp$ or not does not depend on how the nodes of $G$ are labelled with identifiers, and moreover, the usual challenge of \emph{local symmetry breaking} does not arise in the context of decision problems.

Indeed, they conjectured that for any local algorithm $A$ that decides a property $\Pp$ there is an equivalent \emph{Id-oblivious} local algorithm $A^*$ that decides $\Pp$ and that does not use identifiers in the sense that the output of $A^*$ on a node $v\in V(G)$ does not change if we reassign the identifiers, i.e., $A^*(G,\Id,v) = A^*(G,\Id',v)$ for any two assignments $\Id,\Id'\colon V(G)\to\mathbb{N}$.

In this work, we disprove the conjecture. We show that there are graph properties whose local decision requires the output of a constant-time algorithm to depend on the identifier assignment---if the details of the underlying model of distributed computation are set up in a particular way.

\paragraph{Assumptions.}

To understand what our question entails on a technical level, we need to make explicit two critical assumptions about the model of computing.

\emph{Size of identifiers.}
It is commonly assumed that the identifiers are given as $O(\log n)$-bit labels in a graph with $n$ nodes. It is debatable whether it is natural to require bounded identifiers in our case of constant-time algorithms; in any case, we consider both alternatives:
\begin{itemize}[leftmargin=1.5cm,noitemsep]
\item[($\B$)] The size of identifiers is \emph{bounded} by a function of $n$.
\item[($\neg \B$)] The size of identifiers is \emph{unbounded}.
\end{itemize}
Note that, since a local algorithm operates on a graph component-wise, there is no distinction between $(\B)$ and $(\neg\B)$ if we allow all disconnected graphs as input: in either case there will be no bound on $\Id(v)$ as a function of the size of $v$'s component. Thus, in what follows, we work under the promise that the input graph is connected. We will show that whether identifiers help in local decision depends on which of the assumptions $(\B)$ or $(\neg\B)$ we adopt.

\emph{Computability.}
Second, should we restrict the power of local computations? We have two alternatives:
\begin{itemize}[leftmargin=1.5cm,noitemsep]
\item[$(\C)$] The nodes run a \emph{computable} algorithm.
\item[$(\neg\C)$] The nodes can compute any function, possibly \emph{uncomputable}.
\end{itemize}
For many questions in distributed computing, the distinction between $(\C)$ and $(\neg \C)$ is inconsequential and not interesting. However, we will show that whether identifiers help in local decision depends on which of the assumptions $(\C)$ or $(\neg \C)$ we adopt.

\paragraph{Id-oblivious simulation.}
Our results are best motivated by the observation that identifiers are not needed under $(\neg\B,\neg\C)$. Indeed, if $A$ is a $t$-time algorithm deciding a property $\Pp$, we can simulate $A$ by an Id-oblivious $t$-time algorithm $A^*$.

\begin{quote}
{\sffamily Id-oblivious simulation $A^*$:}
For each local neighbourhood $(G',v)$, $G'\subseteq G$, algorithm $A^*$ checks whether there is a local assignment $\Id'\colon V(G')\to\mathbb{N}$ that makes the output $A(G',\Id',v)$ be \emph{no}. If such an assignment exists, we let $A^*$ output \emph{no} on~$v$, too; otherwise, we let $A^*$ output \emph{yes} on $v$.
\end{quote}

We first note that, even though $A^*$ is well-defined, it is not obvious how to compute it, since finding out whether $\Id'$ exists might involve an exhaustive search over an infinite domain. For example, even if $A$ was computable to start with, our $A^*$ is now deciding, a priori, a \emph{computably enumerable} predicate. However, under $(\neg \C)$, this is not a problem.

To see that $A^*$ correctly decides $\Pp$, we note that $A^*$ outputs \emph{no} on some node in $G$, if and only if there is some global assignment $\Id\colon V(G)\to\mathbb{N}$ (i.e., extension of $\Id'$) that makes $A$ output \emph{no} on some node. The identifiers in the assignment $\Id$ may be very large, but under $(\neg\B)$ this is not a problem. Thus, $(G,\Id)$ is a valid input for $A$, and the correctness of $A^*$ now follows from that of $A$.

Our main result in this work is showing that there is no general Id-oblivious simulation in case one of the assumptions $(\B)$ or $(\C)$ is imposed.

\subsection{Our results}

We show that identifiers are necessary in local decision under $(\B)$, and under~$(\C)$.
\begin{theorem} \label{thm:main}
Assume $(\B)$ or $(\C)$. There is a locally decidable property $\Pp$ that cannot be decided with an Id-oblivious local algorithm.
\end{theorem}
In particular, this separates the classes $\LD$ and $\LD^*$ that were previously conjectured to be equal under $(\neg\B,\C)$ by \citet{FHK12}. Here, $\LD$ is the class of locally decidable properties, and $\LD^*\subseteq \LD$ is the class of properties decidable with an Id-oblivious local algorithm.

We prove the separation $\LD^*\neq \LD$ assuming $(\B,\neg\C)$ in Section~\ref{sec:separation-under-b}, and again assuming $(\C)$ in Section~\ref{sec:separation-under-c}. For the latter, more involved separation, we end up using ideas from classical (sequential) computability theory. The use of these techniques should not come as a surprise given that $\LD^*=\LD$ under $(\neg\B,\neg\C)$ as discussed above. We collect the relationships between $\LD^*$ and $\LD$ in the following table:
\begin{center}
\newlength{\colh} \setlength{\colh}{7pt}
\setlength{\arrayrulewidth}{.7pt}
\newcolumntype{C}{>{\centering\arraybackslash}X}
\def\cell#1{\raisebox{-.5\colh}{#1}}
\begin{tabularx}{7cm}{r|CC|l}
\multicolumn{1}{r}{}
& \multicolumn{1}{c}{$(\C)$}
& \multicolumn{1}{c}{$(\neg\C)$} \\[3pt]
\hhline{~|--|}
\cell{$(\B)$} &  \cellcolor[gray]{.9}\cell{$\boldsymbol\neq$} & \cellcolor[gray]{.9}\cell{$\boldsymbol\neq$} & \cell{$\rightarrow$ Section~\ref{sec:separation-under-b}}\\[\colh]
\hhline{~|~|~|}
\cell{$(\neg\B)$} & \cellcolor[gray]{.9}\cell{$\boldsymbol\neq$}  & \cell{$\boldsymbol=$}  \\[\colh]
\hhline{~|-|-|}
\multicolumn{1}{r}{}
\\[-9pt]
\multicolumn{1}{r}{}
& \multicolumn{3}{l}{\hspace{14pt}\raisebox{1pt}{\rotatebox[origin=c]{180}{\large$\Lsh$}} Section~\ref{sec:separation-under-c}}
\end{tabularx}
\end{center}

Finally, we note that the property $\Pp$ that witnesses $\LD\neq\LD^*$ under $(\C)$ becomes decidable with an Id-oblivious algorithm if we allow \emph{randomness}.
\begin{corollary} \label{cor:random}
Property $\Pp$ can be decided (w.h.p.) with an Id-oblivious randomised local algorithm.
\end{corollary}
Randomised local decision was previously studied by~\citet{FKP11,FKPP12}. The corollary above indicates, in particular, that in the Id-oblivious model, the threshold result \cite[Theorem~3.3]{FKP11} that pertains to hereditary languages does not hold if we consider all languages.

\subsection{Local decision in the \texorpdfstring{$\local$}{LOCAL} model}\label{ssec:local-model}

A \emph{labelled graph} is a pair $(G,\inp)$, where $G = (V(G),E(G))$ is a simple undirected graph and function $\inp$ associates a \emph{label} or a \emph{local input}, denoted $\inp(v)$, with each node $v \in V(G)$.

A \emph{labelled graph property} is a collection $\Pp$ of labelled graphs that is invariant under graph isomorphism. That is, if $(G,\inp) \in \Pp$, and $(G',\inp')$ is isomorphic to $(G,\inp)$, then $(G',\inp') \in \Pp$. Examples of labelled graph properties include the following:
\begin{itemize}[noitemsep]
    \item ``proper $3$-colouring'': $(G,\inp) \in \Pp$ if $\inp$ is a proper $3$-colouring of $G$,
    \item ``maximal independent set'': $(G,\inp) \in \Pp$ if the nodes with $\inp(v) = 1$ form a maximal independent set in $G$,
    \item ``planar graphs'': $(G,\inp) \in \Pp$ if $G$ is a planar graph (and $\inp$ is arbitrary).
\end{itemize}
In particular, all graph properties can be interpreted as labelled graph properties. If $\Pp$ is a property, we say that any pair $(G,\inp)\in\Pp$ is a \emph{yes-instance} and any pair $(G,\inp)\notin\Pp$ is a \emph{no-instance}.

An \emph{input} is a triple $(G,\inp,\Id)$, where $(G,\inp)$ is a labelled graph and $\Id\colon V(G)\to\mathbb{N}$ is a one-to-one function. We say that $\Id(v)$ is the \emph{unique identifier} of node $v \in V(G)$.

\paragraph{Local algorithms.}

Let $B(v,t) \subseteq V(G)$ consist of the nodes that are within distance $t$ from $v$ in graph $G$. We write $(G,\inp,\Id) \upharpoonright B(v,t)$ for the restriction of the structure $(G,\inp,\Id)$ to $B(v,t)$.

Let $A$ be a function that associates a \emph{local output} $A(G, \inp, \Id, v) \in \{\yes,\no\}$ with each node $v \in V$ for any input $(G,\inp,\Id)$. We say that $A$ is a \emph{local algorithm} with local horizon $t$ if $A(G, \inp, \Id,v) = A(G', \inp', \Id',v)$ whenever $(G,\inp,\Id) \upharpoonright B(v,t) = (G',\inp',\Id') \upharpoonright B(v,t)$. That is, in a local algorithm the local output of node $v$ depends only on the information that is available in the radius-$t$ neighbourhood of node $v$.

We say that local algorithm $A$ is \emph{Id-oblivious} if $A(G,\inp,\Id,v) = A(G,\inp,\Id',v)$ for any two assignments $\Id,\Id'\colon V(G)\to\mathbb{N}$. That is, renumbering the identifiers does not change the output of an Id-oblivious algorithm. Indeed, we may write the output simply as $A(G,\inp,v)$.

While in the above description we have specified a local algorithm as a function that maps local neighbourhoods to local outputs, we could equally well specify a local algorithm from the perspective of networked state machines that exchange messages with each other: graph $G$ is the structure of the network, each node is a computer, each edge is a communication link, all nodes run the same algorithm, and a node $v \in V(G)$ initially knows only $\inp(v)$ and $\Id(v)$. In essence, a local algorithm with local horizon $t$ is equivalent to a distributed algorithm that runs in $t\pm 1$ synchronous communication rounds in the $\local$ model \cite{L92,PelB00}.

\paragraph{Assumptions.}

Under assumption ($\B$), we assume that there is a function $f$ such that $\Id(v) < f(|V(G)|)$ for any input $(G,\inp,\Id)$.

Under assumption ($\C$), we require that local algorithm $A$ is a computable function of the local neighbourhood. Put otherwise, we require that there is a Turing machine $M_A$ such that for any input $(G, \inp, \Id)$ and any node $v \in G$, given a string that encodes node $v$ and the local neighbourhood $(G,\inp,\Id) \upharpoonright B(v,t)$, machine $M_A$ halts and outputs $A(G, \inp, \Id,v)$.

\paragraph{Local decision.}

Local algorithm $A$ \emph{decides} a property $\Pp$ if the following holds for any input $(G,\inp,\Id)$:
\begin{itemize}[noitemsep]
\item if $(G,\inp)\in\Pp$, then $A(G,\inp,\Id,v) = \yes$ for all $v \in V(G)$,
\item if $(G,\inp)\notin\Pp$, then $A(G,\inp,\Id,v) = \no$ for at least one $v \in V(G)$.
\end{itemize}
If there is a local algorithm that decides $\Pp$, we say that $\Pp$ is in class $\LD$. If there is an Id-oblivious local algorithm that decides $\Pp$, we say that $\Pp$ is in class $\LD^*$.

\paragraph{Promise problems.}

While our constructions do not make use of promise problems, we will refer to them in some introductory examples. If we say that we have \emph{promise} $\Pp'$, then we are only interested in inputs $(G,\inp,\Id)$ with $(G,\inp) \in \Pp'$.

In particular, if $(G,\inp,\Id)$ is an input that violates the promise, we do not put any requirements on $A(G,\inp,\Id,v)$. Even if we work under assumption ($\C$), we do not require that machine $M_A$ halts for inputs that violate the promise. Put otherwise, $A$ can be a partial function, undefined for inputs that violate the promise.

\subsection{Related work}

The question of how to locally decide (or verify) languages has been gaining attention in recent years \cite{DHKKNPPW,FHK12,FKP11, GS11,KK07,KKM11,KKP10,KKP11}. Inspired by traditional computational complexity theory, \citet{FKP11} suggested that the study of decision problems may provide new structural insights also in the distributed computing setting. While the original focus was on the $\local$ model, recent work has taken the first steps towards a computational complexity theory in various other contexts of distributed computing \cite{FP12,FRT11,FRT12}.

\paragraph{Local decision.}

The classes $\LD$, $\NLD$ and $\BPLD$ defined by \citet{FKP11} are the distributed analogues of the classes $\P$, $\NP$ and $\BPP$, respectively. The paper \cite{FKP11} provides structural results, develops a notion of local reduction, and establishes completeness results. One of the main results is that, for a large class of languages, called \emph{hereditary languages}, there exists a sharp threshold for randomisation, above which randomisation does not help.

\paragraph{Identifiers and local decision.}

More recently, \citet{FHK12} defined the \emph{Id-oblivious} model, and the corresponding class of languages $\LD^*$, aiming to understand better the role of identities in local decision. They also conjectured that $\LD^*=\LD$. Informally, the conjecture states that for constant time computations, identities do not play any role except for allowing nodes to identify their local neighbourhoods.

Several positive evidences where given supporting this conjecture~\cite{FHK12}. Specifically, it is shown that $\LD^*=\LD$ holds for hereditary languages and languages defined on paths, with a finite set of input values. Moreover, it was shown that equality holds in the non-deterministic setting, i.e., $\NLD^*=\NLD$.

\paragraph{Identifiers and local construction.}

The role of identifiers is different in local algorithms that need to \emph{construct} a solution. From the perspective of construction tasks, it is easy to see that the usual $\local$ model is much stronger than the Id-oblivious model: there are many tasks that are trivial in $\local$ and impossible to solve with an Id-oblivious algorithm (examples: finding an orientation of the edges; 2-colouring a 1-regular graph).

Therefore to ask meaningful questions related to the role of unique identifiers in construction tasks, we usually compare the $\local$ model with models that retain some symmetry-breaking information---two such models are $\OI$, \emph{order-invariant algorithms}, and $\PO$, \emph{port numbering and orientation}.
\begin{itemize}[noitemsep]
    \item In the $\OI$ model \cite{NS93}, the output of an algorithm is not allowed to change if we reassign the identifier while preserving their relative order.
    \item In the $\PO$ model \cite{mayer95local}, there is an ordering on the incident edges, and all edges carry an orientation.
\end{itemize}
Note that model $\OI$ is stronger than the Id-oblivious model: in the Id-oblivious model, $A^*(G,\Id,v) = A^*(G,\Id',v)$ for \emph{any} two assignments $\Id,\Id'\colon V(G)\to\mathbb{N}$, while in the $\OI$ model, we only require this for assignments $\Id,\Id'\colon V(G)\to\mathbb{N}$ that satisfy $\Id(u) < \Id(v) \iff \Id'(u) < \Id'(v)$. This difference makes the $\OI$ model much stronger.

Indeed, it turns out that from the perspective of strictly local algorithms, for many graph problems models $\local$ and $\OI$ are equally strong: \citet{NS93} prove that for problems whose \emph{decision} version can be solved locally, \emph{construction} is possible in $\local$ if and only if it is possible in $\OI$. More recently, \citet{GHS12} shows that there is also a general class of \emph{optimisation} problems for which $\local$, $\OI$ and $\PO$ are equally expressive.

The results of \citet{NS93} and \citet{GHS12} focus on bounded-degree graphs. They also make a subtle technical assumption: each node produces a local output from a constant-size set. This is necessary: \citet{HHRS12} give an example of a natural problem that violates this assumption---and separates $\local$ and $\OI$.

\paragraph{Bounds on $\bm{n}$.}

It turns out that in decision problems, unique identifiers are helpful for one reason, and for one reason only: obtaining an estimate on $n$, the number of nodes. Indeed, by prior work we already know that $\LD^*=\LD$ holds assuming that every node knows an upper bound on the total number of nodes in the input graph~\cite{FHK12}.

Of course we can interpret a decision problem as a very special kind of construction problem, and therefore the present work also shows that some construction problems can exploit numerical identifiers to learn about $n$. However, this is a highly atypical example. For classical graph problems this information does not help a local algorithm---the identifiers are typically used for local symmetry breaking and their numerical magnitude is inconsequential.

However, if we step outside the field of strictly local algorithms, it is common to \emph{assume} that all nodes know the same upper bound on $n$. This is a convenient assumption that often simplifies algorithm design. \citet{KSV11} show that in many cases it is merely a convenience---the knowledge of an upper bound on $n$ is not essential.

\section{Separation under bounded identifiers} \label{sec:separation-under-b}

In this section we work under assumption $(\B,\neg\C)$ and exhibit a locally decidable property $\Pp$ that cannot be decided with an Id-oblivious local algorithm.

Let $f\colon \mathbb{N}\to\mathbb{N}$ be such that $\Id(v)< f(n)$ for all $v\in V(G)$, where $G$ is a connected input graph. The reason identifiers are useful is that they leak information about $n$. For example, if a node is given an identifier $i$, it can deduce that $n > f^{-1}(i)$, where we denote by $f^{-1}(i)$ the smallest $j$ such that $f(j)\geq i$.

\paragraph{Promise problem.}
As an illustration, we first describe a simple promise problem in $\LD\smallsetminus\LD^*$.
\begin{quote}
{\sffamily Promise problem:}
The instances are labelled graphs $(G,r)$ where $G$ is an $n$-cycle and $r\in\mathbb{N}$ is a constant input label. We promise that either $n = r$ or $n = f(r)$.

We have a \emph{yes}-instance if $n = r$ and a \emph{no}-instance if $n = f(r)$.
\end{quote}
Note that $r$-cycles and $f(r)$-cycles cannot be told apart by an Id-oblivious algorithm as they are locally indistinguishable topology-wise when $r$ is large. However, we can solve the problem using identifiers: the $f(r)$-cycles can be rejected, because there is a node with identifier at least $f(r)$, which is too large to be found in the $r$-cycle. (We can exploit assumption $(\neg\C)$ here if $f$ is uncomputable.)

It is not much harder to design a promise-free example in $\LD\smallsetminus\LD^*$---we do this next.

\paragraph{Promise-free problem.}
Define $R(r) := f(2^{r+1} + 1)$. The key idea is that
\begin{itemize}[noitemsep]
    \item if the instance is a complete depth-$r$ binary tree, all identifiers are smaller than $R(r)$,
    \item if the instance is a complete depth-$R(r)$ binary tree, there is an identifier at least~$R(r)$.
\end{itemize}
Intuitively, we can use identifiers to accept ``small'' instances and reject ``large'' instances. The nontrivial part is to make sure that we can also reject instances that are neither small nor large.

A \emph{layered depth-$k$ tree} is a complete binary tree of depth $k$ where, in addition, nodes at each level are connected by a path in the natural order; see Figure~\ref{fig:trees}. Denote by $T_r$ the labelled graph consisting of a layered depth-$R(r)$ tree. Each node of $T_r$ is labelled with $(r,x,y)$, where the coordinates $(x,y)$ indicate the position of the node in the binary tree.

\begin{figure}[t]
    \centering
    \makebox[0pt][c]{\includegraphics[page=3]{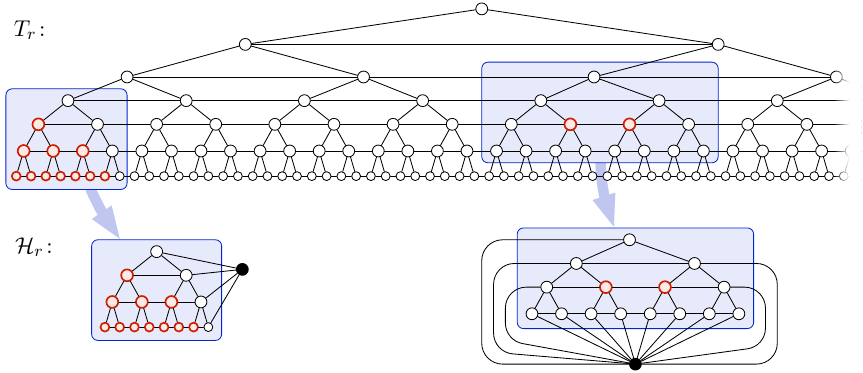}}
    \caption{Graph $T_r$ is a layered tree of depth $R(r) \gg r$. Each graph $H^+ \in \cH_r$ is a layered tree of depth $r$, augmented with a single pivot node (black). The nodes that are far from the boundary (highlighted) have local neighbourhoods that are indistinguishable from the local neighbourhood of the same node in $T_r$.}\label{fig:trees}
\end{figure}

Write $H\leq_r T_r$ if a labelled graph $H$ is an induced subgraph of the labelled graph $T_r$, and the topology of $H$ is a layered depth-$r$ tree. Call $u\in V(H)$ a \emph{border node} if $u$ has a neighbour in $V(T_r)\smallsetminus V(H)$. We define $H^+$ to be $H$ together with a new node (\emph{pivot node}) that is adjacent to all the border nodes of $H$; see Figure~\ref{fig:trees}. We collect $\cH_r :=\{H^+ : H \leq_r T_r \}$. We are now ready to define
\[
    \Pp := \bigcup_{r\geq 0} \cH_r, 
    \qquad
    \Pp' := \Pp \cup \{ T_r : r \geq 0 \}.
\]
We will refer to labelled graphs in $\Pp$ as ``small'' instances and graphs in $\Pp' \smallsetminus \Pp$ as ``large'' instances. Of course instances of $\Pp$ are only small in comparison with the parameter $r$ that is encoded in the labelling of the graph; we have arbitrarily large graphs in both $\Pp$ and $\Pp'$.

We will next show that the construction satisfies the following properties:
\begin{itemize}[noitemsep]
    \item $\Pp' \in \LD^*$, that is, even if we do not have access to unique identifiers, we can verify that the input is \emph{either} small \emph{or} large. Hence we do not need to rely on a promise---we can locally verify it.
    \item $\Pp \in \LD$, that is, we can reject large instances with the help of identifiers,
    \item $\Pp \notin \LD^*$, that is, we cannot distinguish between small and large instances with Id-oblivious algorithms.
\end{itemize}

$(\Pp'\in\LD^*)$:
The overall structure of a layered depth-$R(r)$ tree is straightforward to verify locally with the help of coordinates; we can also easily check that all nodes agree on the value of~$r$. We can verify that the coordinates satisfy $0 \le x < 2^y$ and $0 \le y \le R(r)$, there is no parent iff $y = 0$, there are no children iff $y = R(r)$, etc.

The non-trivial part is the case of a pivot node. The crucial property is that a pivot node sees \emph{all} border nodes of a small instance. Therefore a pivot node can verify that the size of the border (as well as the coordinates of the border nodes) agree with the definition of a small instance.

In essence, if we encounter a pivot node, we must have a small instance: if we fix the structure near the border nodes, and then complete it so that it is locally consistent with the structure of a layered tree, we will arrive at a labelled graph in $\Pp$. On the other hand, if we never encounter a pivot node, we must have a large instance.

$(\Pp\notin\LD^*)$:
Suppose for contradiction that $A^*$ is a $t$-time Id-oblivious algorithm that decides~$\Pp$. For a large enough $r\gg t$, we have that each $t$-neighbourhood in $T_r$ is already found in one of the \emph{yes}-instances in $\cH_r$. But because $A^*$ accepts all of $\cH_r$, it must also accept the \emph{no}-instance $T_r$, which is a contradiction.

$(\Pp\in\LD)$:
The only difficulty in locally deciding $\Pp$ is to be able to reject $T_r$ while accepting all graphs in $\cH_r$. But there is a node in $T_r$ with an identifier at least~$R(r)$, which is too large to be found in the graphs $\cH_r$.

\section{Separation under computability} \label{sec:separation-under-c}

In this section we assume that all local algorithms are computable $(\C)$. We will exhibit a locally decidable property $\Pp$ that cannot be decided by an Id-oblivious local algorithm.

\paragraph{Promise problem.}
Again, to illustrate our approach, we first describe a simple promise problem that separates $\LD^*$ and $\LD$.
\begin{quote}
{\sffamily Promise problem $\Rr$:} The instances are labelled graphs $(G,M)$ such that $G$ is an $n$-cycle; the constant input label $M$ is a Turing machine; and if $M$ halts in exactly $s$ steps (when started on a blank tape) then we promise that $n \geq s$.

We have a \emph{yes}-instance if $M$ runs forever and a \emph{no}-instance if $M$ halts.
\end{quote}

$(\mathcal{R}\in\LD)$: The problem $\Rr$ is locally decidable using identifiers. Indeed, a node with identifier $i$ first simulates $M$ for $i$ steps. Then, if $M$ stops within this many steps, we output \emph{no}; otherwise we output \emph{yes}. For correctness, note that our promise implies that for every \emph{no}-instance $(G,M)$ where $M$ halts, there will be some node $v$ with identifier at least as large as $M$'s run-time, and $v$ will be able to reject $(G,M)$.

$(\mathcal{R}\notin\LD^*)$: On the other hand, it is easy to see that any Id-oblivious algorithm for $\Rr$ has to solve the halting problem without the additional knowledge of $M$'s run-time, which is an uncomputable task.

In this section, our goal is to construct a promise-free version of this decision problem.

\subsection{Overview}
The computationally difficult part in our decision problem $\Pp$ will be to determine whether a given Turing machine $M$ halts and outputs $0$ (when started on a blank tape).

To make $\Pp$ easy for an algorithm using identifiers, we will require that the instance $G$ contains a grid-like locally checkable execution table of $M$. This way---as in the promise problem example---there will be some node $v$ that has an identifier larger than $M$'s run-time. The node $v$ can then locally simulate $M$ to discover its output.

To make $\Pp$ hard for an Id-oblivious algorithm, we need to obfuscate the structure of $G$ so that its local topology does not reveal any useful information about the execution of $M$. In particular, even if $M$ halts, no local neighbourhood of $G$ should certify this fact. This way, an Id-oblivious algorithm is left with trying to find out $M$'s output without any additional means. More formally, such an algorithm would need to separate the languages
\[
L_i := \{ M : M\ \text{outputs}\ i\},\qquad i=0,1,
\]
which is known to be impossible for a computable function:
\begin{lemma}[{e.g.~\cite[p.\ 65]{papadimitriou94computational}}] \label{lem:inseparability}
The languages $L_0$ and $L_1$ are \emph{computably inseparable}, i.e., there is no computable set $R$ such that $L_0 \subseteq R$ and $L_1\cap R = \varnothing$. \qed
\end{lemma}

\paragraph{Implementation.}
For a pair $(M,r)$, where $M$ halts and $r\in\mathbb{N}$ is a locality parameter, we will construct a graph $G(M,r)$ satisfying the following properties.
\begin{enumerate}[label=(P\arabic*),noitemsep]
\item The execution table of $M$ is contained in $G(M,r)$.
\item It is locally decidable (even in $\LD^*$) whether an instance is of the form $G(M,r)$.
\item The $r$-neighbourhoods of $G(M,r)$ reveal only computable information about $M$. More formally, there is an algorithm $B$ that halts on all inputs $(N,r)$, where $N$ is any Turing machine, and outputs a finite set of $r$-neighbourhoods $B(N,r)$ such that
\[\normalfont
\text{$N$ halts}\quad\Longrightarrow\quad B(N,r) = \{\ \text{$r$-neighbourhoods of $G(N,r)$}\ \}.
\]
Note, especially, that $B$ halts even if $N$ does not!
\end{enumerate}
Suppose for a moment that we have a construction satisfying (P1--P3). We can now define
\[
\Pp := \{ G(M,r) : M \text{ outputs } 0\}.
\]
\begin{theorem}
$\Pp\in\LD\smallsetminus\LD^*$ under $(\C)$.
\end{theorem}
\begin{proof}
$(\Pp\in\LD)$: Given $(G,\Id)$ as input, a node $v\in V(G)$ computes in two stages. First, $v$ performs its local test according to (P2) to see if $G=G(M,r)$ for some $(M,r)$. If this test fails, $v$ outputs \emph{no}. Otherwise $v$ proceeds to the second stage where $v$ locally simulates $M$ for $\Id(v)$ steps. If the simulation finishes and $M$ outputs something other than $0$, then $v$ outputs \emph{no}; otherwise $v$ outputs \emph{yes}.

For correctness, we need only note that in case all nodes pass the first stage, we have that $G=G(M,r)$, and thus, by (P1), there will be some node $v$ with so large an identifier that $v$ will finish the simulation of $M$ in the second stage and discover $M$'s true output.

$(\Pp\notin\LD^*)$: For the sake of contradiction, suppose that an Id-oblivious algorithm $A^*$ with run-time $t$ decides $\Pp$. We show how $A^*$ can be exploited to separate the languages $L_0$ and $L_1$.
\begin{quote}
{\sffamily Separation algorithm $R$:}
Given a Turing machine $N$ we first compute $B(N,t)$. Then, we run $A^*$ on all the $t$-neighbourhoods in $B(N,t)$. We accept $N$ precisely if $A^*$ accepts all of $B(N,t)$.
\end{quote}
First, note that, by (P3), our algorithm $R$ halts on every input $N$. Moreover, suppose that $N$ halts. Then $R$ accepts $N$ iff $A^*$ accepts every $t$-neighbourhood of $G(N,r)$ iff $A^*$ accepts $G(N,r)$ iff $G(N,r)\in\Pp$ iff $N$ outputs $0$. But this contradicts Lemma~\ref{lem:inseparability}.
\end{proof}

Indeed, it remains to give the details of a construction satisfying (P1--P3).

\subsection{Construction of \texorpdfstring{$\bm{G(M,r)}$}{G(M,r)}} \label{ssec:construction}

Let $M$ be a Turing machine that halts. Each node in the graph $G=G(M,r)$ will have $(M,r)$ as part of their input labelling. The graph $G$ will consist of two parts:
\begin{itemize}[noitemsep]
\item the \emph{execution table} $T$ of $M$, and
\item a certain \emph{fragment collection} $\cC$.
\end{itemize}
See Figure~\ref{fig:GMr}.

\begin{figure}
    \centering
    \makebox[0pt][c]{\includegraphics[page=1]{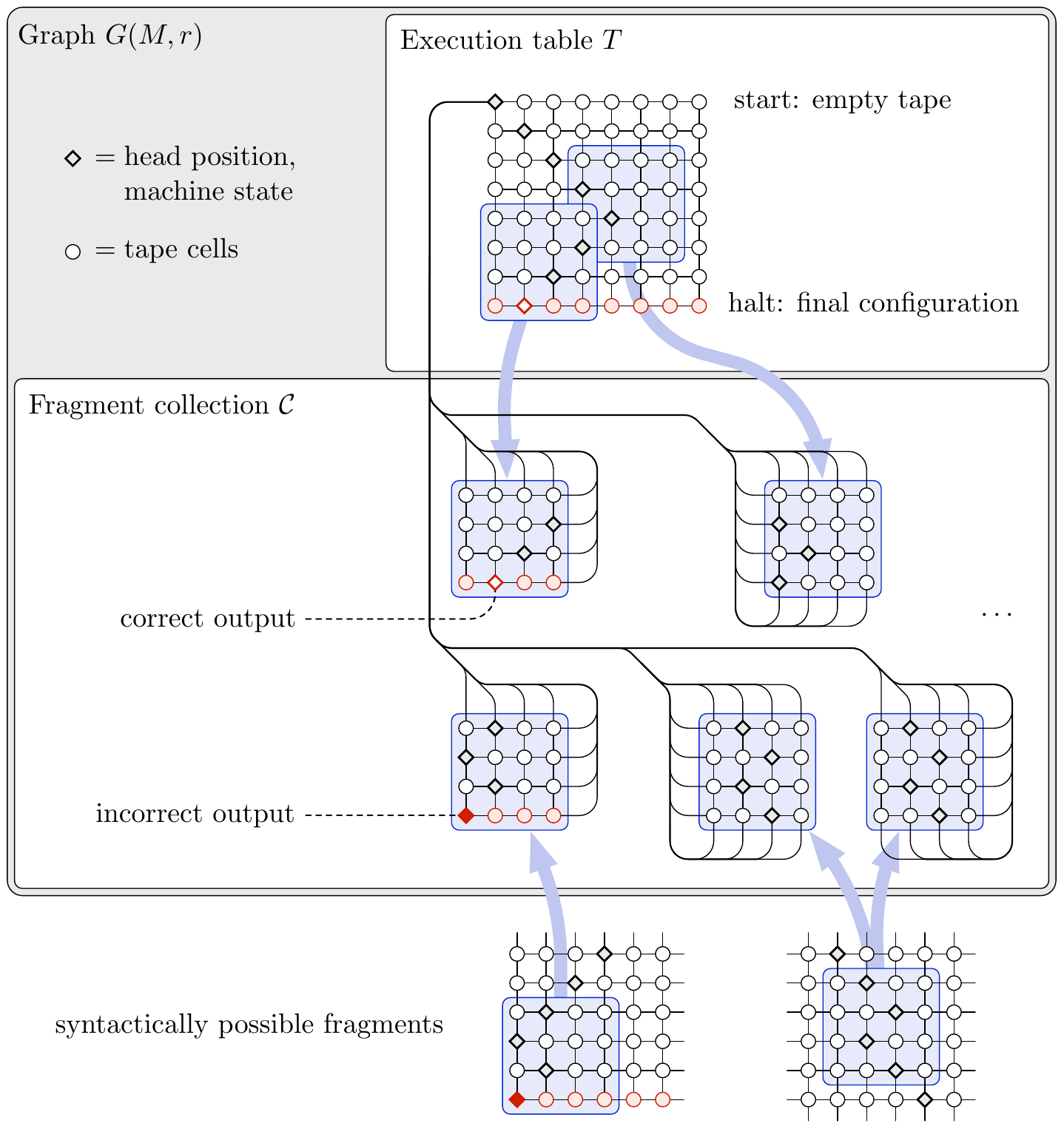}}
    \caption{Construction of graph $G(M,r)$.}\label{fig:GMr}
\end{figure}

\paragraph{Execution table.}
Let $s$ be the running time of $M$.
The execution table $T$ of $M$ will be represented, as per usual, as a labelled square grid graph on nodes $[s+1]\times[s+1]$, where two nodes are adjacent if their Euclidean distance is $1$. We think of the edges of $T$ as being oriented from top to bottom and from left to right. Such an orientation can be locally supplied by labelling $(x,y)$ with $(x\bmod3,\,y\bmod3)$.

\emph{Labels for execution.}
The $i$-th row of $T$ corresponds to the configuration of $M$ before the $i$-th step of the execution: the nodes are labelled with tape cell contents, and the read-write head of the machine is owned by exactly one node per row; this node also records the state of the machine. The first row contains just blank symbols, and the computation starts with the head on the leftmost node, which we call the \emph{pivot node}.

The exact details of this labelling scheme are not important. Any reasonable scheme will do. We only require that the size of the labels is bounded by a computable function of $M$. For example, we cannot allow the nodes on the $i$-th row to hold the number $i$ in their labels, since, intuitively, this would leak information about $M$'s run-time to an Id-oblivious algorithm. (More precisely, this would mess up our construction of $\cC$ below.)

\emph{Local decidability.}
It is well known that valid executions of a Turing machine can be checked locally---at least once we somehow know that the instance is really a labelled square grid and not, e.g., a torus-like graph that locally looks like a grid. To make $T$ locally checkable, we need to augment it with some special structure; we take care of this technicality in Appendix~\ref{app:details}.

\paragraph{Fragment collection.}
The purpose of the fragment collection $\cC$ is to ensure property (P3).

\emph{Intuition.}
If we had $G=T$, an Id-oblivious algorithm could decide whether $M$ output $0$ simply by checking if there was a local neighbourhood in $G=T$ where $M$ is in a halting state with output $0$.

To prevent this from happening, we add superfluous table fragments to $G$. In fact, we will let $G$ contain all syntactically possible execution table fragments. This way, the answer to the question ``Does there exists a local neighbourhood in $G$ where $M$ is in such-and-such a state'' will always be \emph{yes}. In effect, when an Id-oblivious algorithm is exploring $G$ locally, it learns nothing about the execution of $M$ that it could not compute by itself.

\emph{Construction.}
Let $F$ be a $3r \times 3r$ grid graph. Consider labelling $F$ in all possible ways that satisfy the local consistency rules of $T$. That is, we put no limitations on how the boundary nodes are labelled, as long as
\begin{itemize}[noitemsep]
\item the $(\text{mod}\ 3)$-labels give a consistent orientation, and
\item every $2\times 2$ sub-table of $F$ is consistent with the transition function of $M$. 
\end{itemize}
We let $\cC=\cC(M,r)$ consist of these labelled versions of $F$.

The important property here is that every $r$-neighbourhood in $T$ (including those near a boundary of $T$) is found already in some labelled fragment in $\cC$.

\emph{Efficiency.}
The construction of $\cC$ is purely syntactic: for any machine $N$ (that does not necessarily halt), we can efficiently generate $\cC(N,r)$ by a simple enumeration of all possible labellings, as our labelling scheme uses bounded labels. We record this observation.
\begin{lemma} \label{lem:fragments}
There is an algorithm that on input $(N,r)$ outputs the finite collection $\cC(N,r)$. \qed
\end{lemma}

\paragraph{Putting $\bm{G}$ together.} To construct $G$ we glue together $T$ and the fragments $\cC$. Details follow.

\emph{Natural borders.}
Consider the leftmost column of nodes $C$ in a labelled fragment $F\in\cC$. We call $C$ a \emph{natural border} if $C$ could, in principle, appear on the leftmost column of an execution table of $M$, i.e., if the machine head never moves to, or appears from, the left of $C$. We say that the rightmost column is \emph{natural} under analogous circumstances. The bottom row is \emph{natural} if it does not contain the machine head in a non-halting state. The top row is never natural.

Here is a technical point: we need the non-natural borders to always form a connected subgraph of $F$. The only situation where this is currently violated is when precisely the top and bottom rows of $F$ are non-natural, but this is easily fixed by replacing $F$ with two of its variants where the left and right borders are interpreted non-natural in turn. We now gain the following property, which becomes useful when proving that $G$ is locally decidable.
\begin{quote}
{\sffamily Border property:} Given a subgraph induced on the non-natural borders of a fragment~$F\in \cC$, the local transition rules of $M$ reconstruct $F$ uniquely.
\end{quote}

\emph{Construction.}
The graph $G$ consists of (i)~the table $T$, (ii)~the fragments $\cC$, and also (iii)~new edges that connect each node of a non-natural border in $\cC$ to the pivot node of $T$.

This completes the description of $G$. We leave the straightforward but tedious details of checking that $G$ is locally decidable to Appendix~\ref{app:details}.

\emph{Efficiency.}
Finally, for the purposes of (P3), we note that our construction of $G(M,r)$ is highly explicit in the sense that the set of $r$-neighbourhoods of $G(M,r)$ can be computed even without the knowledge of $M$ halting. 
\begin{quote}
{\sffamily Neighbourhood generator $B$:}
On input $(N,r)$, where $N$ does not necessarily halt, we first compute $\cC=\cC(N,r)$ using Lemma~\ref{lem:fragments}. Then, we begin constructing the (possibly infinite) computation table $T$ of $N$ for some $4r$ rows, each of width $4r$; call the resulting table fragment $T_{4r}\subseteq T$. We then glue $\cC$ to the pivot of $T_{4r}$ as described above to obtain a graph $G_{4r}$. Finally, we output the set of $r$-neighbourhoods in $G_{4r}$ that do not contain nodes from the bottom row of $T_{4r}$.
\end{quote}
The correctness of $B$ follows from the observation that, if $N$ halts, every $r$-neighbourhood in $G(N,r)$ is already found in $G_{4r}$. This establishes property (P3) and completes our proof.

\subsection{Randomisation helps an Id-oblivious algorithm}

To conclude this section, we point to another application of our property $\Pp$, this time in the setting of \emph{randomised} local decision. Namely, we observe that $\Pp$ can be decided by an Id-oblivious algorithm if and only if we allow randomness.

A \emph{randomised} local algorithm has access to an unbounded string of random bits. For $p,q\in(0,1]$, we say that a randomised local algorithm $A$ is a {\em $(p,q)$-decider} for $\Pp$ if the following holds for any input $(G,\inp,\Id)$:
\begin{itemize}[noitemsep]
\item if $(G,\inp)\in\Pp$, then $A(G,\inp,\Id,v) = \yes$ for all $v \in V(G)$ with probability at least~$p$,
\item if $(G,\inp)\notin\Pp$, then $A(G,\inp,\Id,v) = \no$ for at least one $v \in V(G)$ with probability at least~$q$.
\end{itemize}
The power of randomness is still lacking a full characterisation in the context of local decision~\cite{FKP11,FKPP12}.

\paragraph{Randomised Id-oblivious decider for $\bm{\mathcal{P}}$.}
Even though an Id-oblivious algorithm cannot use randomness to generate a fresh set of globally unique identifiers without any knowledge of~$n$, we can still generate a few large numbers with high probability. This suffices for deciding $\Pp$ without identifiers, since, in addition to (P2), we only need some node $v$ to obtain a number $n_v \geq n$ so that $v$ can finish simulating $M$ in $n_v$ steps.

To this end, we let a node $v$ toss a coin repeatedly until a head occurs, say after $\ell_v$ tosses. We set $n_v := 4^{\ell_v}$. The probability that no node has $n_v\geq n$ is then
\[
\Pr[\,\forall v\colon n_v < n\,] \leq (1-1/\sqrt{n})^n = o(1).
\]
That is, with probability at least $1-o(1)$ we can reject an instance $G(M,r)$ where $M$ halts with output other than $0$. Hence, we obtain an Id-oblivious $(1,1-o(1))$-decider for $\Pp$.

This proves Corollary~\ref{cor:random}.

\DeclareUrlCommand{\Doi}{\urlstyle{same}}
\renewcommand{\doi}[1]{\href{http://dx.doi.org/#1}{\footnotesize\sf doi:\Doi{#1}}}
\bibliographystyle{plainnat}
\bibliography{ld-id,medium}

\appendix

\section{Construction details} \label{app:details}

In this appendix we present the details that were skipped in Section~\ref{ssec:construction}.

\paragraph{Pyramidal execution table.}
We describe how to augment the execution table $T$ of $M$ so that it becomes locally checkable. For clarity of exposition, we assume that $s+1$ is a power of~$2$, say $s+1=2^h$ for some $h$---this assumption is easy to remove by modifying the following constructions slightly.

Denote the node set of $T$ by $[2^h]\times[2^h]\times\{0\}$. We use an idea from Section~\ref{sec:separation-under-b}: we attach a pyramid-shaped \emph{layered quadtree} on top of~$T$. That is, let $\widehat{T}$ be the graph that is arranged in layers $z=0,1,\ldots,h$ such that $T$ makes up the $0$-th level; the $z$-th level contains a square grid on nodes $[2^{h-z}]\times[2^{h-z}]\times\{z\}$; and each node $(x,y,z)$ on level $z\leq h-1$ is connected to $(\lceil x/2\rceil,\lceil y/2\rceil,z+1)$ on level $z+1$; see Figure~\ref{fig:pyramid}. The new nodes $V(\widehat{T}) \smallsetminus V(T)$ do not receive labels, except, of course, the universal label $(M,r)$.

\begin{figure}[b]
    \centering
    \includegraphics[page=2]{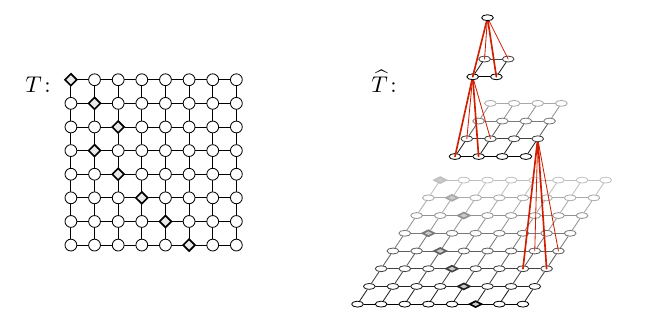}
    \caption{Table~$T$ and pyramid~$\widehat{T}$.}\label{fig:pyramid}
\end{figure}

\paragraph{Pyramidal fragments.}
Since our construction is now going to use the pyramidal $\widehat{T}$ instead of $T$, we need to adjust our definition of the table fragments $\cC$ accordingly. Analogously, we consider the pyramidal versions the fragments in $\cC$:
\[
\widehat{\cC} := \{ \widehat{F} : F\in\cC\}.
\]

However, since attaching a pyramid on top of a fragment decreases shortest-path distances between nodes, we need to use larger fragments than in Section~\ref{ssec:construction}. To fool an $r$-time algorithm, it is sufficient that the pyramids $\widehat{F}$ have height $3r$ (i.e., grid-size is $2^{3r}\times 2^{3r}$). This way we recover the critical property: each $r$-neighbourhood that could syntactically arise in $\widehat{T}$ can already be found in $\widehat{\cC}$.

The graph $G(M,r)$ is then defined similarly as in Section~\ref{ssec:construction}: we glue the fragments $\widehat{\cC}$ to the pivot of $\widehat{T}$ by their non-natural borders.

Note also that in verifying the property (P3) we now need the neighbourhood generator $B$ to first construct a sub-table $T_R\subseteq T$ containing some $R=2^{4r}$ initial rows and columns, and then glue $\widehat{\cC}$ and $\widehat{T}_R$ together.

\paragraph{$\bm{G(M,r)}$ is locally decidable.}
Suppose we are given an instance $G$; we argue how to locally decide (even in $\LD^*$) whether $G=G(M,r)$ for some $(M,r)$.

\begin{enumerate}
\item
All nodes first make sure they are given the same pair $(M,r)$ as part of their local input.
\item
Each node in $G$ should then belong to a layered quadtree. By design, the structure of a quadtree is such that the nodes can locally tell apart adjacent layers and recognise the inter-layer edges. In particular, each pyramid has a unique top node, which fixes its global structure.

If the general quadtree structure is consistent, we can ignore all but the bottommost layer of each pyramid, and be convinced that $G$ consists of square grids that are connected together by some inter-grid edges.
\item
The labelling inside each grid should follow the local execution rules of $M$. Also, we should have a consistent orientation on each grid.
\item
The border nodes of a grid can collectively verify that the grid is either \emph{fragment-like} (all nodes in the topmost row are incident to inter-grid edges) or a full execution table (the top-left node is the only node incident to inter-grid edges).
\item
All top-left grid corners should see at least one \emph{pivot candidate} $v$ that is part of a full execution table. But we can impose that any such $v$ is globally unique:
\begin{itemize}
\item
First, $v$'s own execution table, call it $T$, cannot have any other nodes with outgoing inter-grid edges assuming that all nodes in $T$ pass steps 3 and 4.
\item
Second, consider the grids $\cC$ that adjoin $v$. Node $v$ can check that each grid in $\cC$ has fragment-like non-natural borders. In particular, we can check that the non-natural borders form a connected subgraph in each grid---if the bottom row of a grid is non-natural, it is sufficient to verify that one of the side borders is also non-natural. But then, exploiting the \emph{Border property} from Section~\ref{ssec:construction}, $v$ can figure out the exact structure of $\cC$ provided the nodes in $\cC$ have passed step 3. It follows that there are no inter-grid edges unseen by $v$.
\end{itemize}
This establishes the uniqueness of $v$.
\item
Finally, $v$ can check that $\cC = \cC(M,r)$ using Lemma~\ref{lem:fragments}.
\end{enumerate}

\end{document}